\documentclass[10pt,journal,compsoc]{IEEEtran}
\ifCLASSOPTIONcompsoc
  \usepackage[nocompress]{cite}
\else
  \usepackage{cite}
\fi
\ifCLASSINFOpdf
\else
\fi
\usepackage{xcolor}
\usepackage[vlined,linesnumbered,ruled]{algorithm2e}

\usepackage{stmaryrd}
\usepackage{mathtools} 
\usepackage{nicefrac} 
\usepackage{wasysym}					
\usepackage{url} 
\usepackage{microtype} 
\usepackage{xspace} 
\usepackage{csquotes} 
\usepackage{enumitem}
\usepackage{siunitx}
\usepackage{booktabs}
\usepackage{pgfplots}
\usepackage{balance}
\usepackage{tikz}
\usepackage{tikzscale} 
\usetikzlibrary{positioning,backgrounds,scopes,fit, arrows,positioning, calc}
\usepackage[english]{babel}
\usepackage{csquotes}
\usepackage{microtype}
\usepackage{blindtext}
\usepackage{pdfpages}
\usepackage{balance}

\usepackage{memhfixc}
\usepackage{hyphenat}
\usepackage{newunicodechar} 
\usepackage{siunitx} 
\usepackage{multirow}
\usepackage[mathscr]{euscript}
\usepackage{amsmath}
\usepackage{amssymb}
\usepackage{amsthm}
\usepackage{cleveref}

\SetKwInput{KwInput}{Input}                
\SetKwInput{KwOutput}{Output}              

\DeclareMathOperator{\thespan}{span}

\newcommand{\A}{\ensuremath{\mathbf{A}}\xspace} 
\renewcommand{\a}{\ensuremath{\mathbf{a}}\xspace} 

\renewcommand{\H}{\ensuremath{\mathbf{H}}\xspace}
\newcommand{\U}{\ensuremath{\mathbf{U}}\xspace}
\newcommand{\V}{\ensuremath{\mathbf{V}}\xspace}
\renewcommand{\u}{\ensuremath{\mathbf{u}}\xspace}

\newcommand{\y}{\ensuremath{\mathbf{y}}\xspace}
\newcommand{\x}{\ensuremath{\mathbf{x}}\xspace}
\newcommand{\Q}{\ensuremath{\mathbf{Q}}\xspace}
\newcommand{\q}{\ensuremath{\mathbf{q}}\xspace}
\newcommand{\R}{\ensuremath{\mathbf{R}}\xspace}

\renewcommand{\r}{\ensuremath{\mathbf{r}}\xspace}
\renewcommand{\b}{\ensuremath{\mathbf{b}}\xspace}
\renewcommand{\S}{\ensuremath{\mathbf{S}}\xspace}

\newcommand{\e}{\ensuremath{\mathbf{e}}\xspace}

\newcommand{\svd}{\ensuremath{\operatorname{\mathtt{SVD}}}\xspace}

\newcommand{\toAgg}{\ensuremath{\operatorname{\mathtt{send-to-aggregator}}}\xspace}
\newcommand{\orthonormalize}{\ensuremath{\operatorname{\mathtt{orthonormalize}}}\xspace}
\newcommand{\fedOrthonormalize}{\ensuremath{\operatorname{\mathtt{federated-gram-schmidt}}}\xspace}
\newcommand{\fedNorm}{\ensuremath{\operatorname{\mathtt{federated-centering}}}\xspace}

\newcommand{\toClient}{\ensuremath{\operatorname{\mathtt{send-to-client}}}\xspace}

\newcommand{\stackVert}{\ensuremath{\operatorname{\mathtt{stack-vertically}}}\xspace}

\newcommand{\householder}{\ensuremath{\operatorname{\mathtt{householder-reflection}}}\xspace}
\newcommand{\givens}{\ensuremath{\operatorname{\mathtt{compute-givens-parameter}}}\xspace}

\newcommand{\ie}{i.\,e.\@\xspace}
\newcommand{\eg}{e.\,g.\@\xspace}

\newcommand{\wrt}{w.\,r.\,t.\@\xspace}

\newcommand{\cf}{cf.\@\xspace}

\newtheorem{proposition}{Proposition}[section]

\begin{document}
	
	\title{Privacy of federated QR decomposition using additive secure multiparty computation}
	
	\author{Anne~Hartebrodt,~\IEEEmembership{Member,~IEEE,}
        Richard~R\"ottger
\IEEEcompsocitemizethanks{\IEEEcompsocthanksitem A.H. and R.R are with University of Southern Denmark\protect\\
E-mail: hartebrodt@imada.sdu.dk\protect\\
Orcid: A.H. 0000-0002-9172-3137; R.R. 0000-0003-4490-5947}
\thanks{Manuscript received; revised }}



	
	
	
	\IEEEtitleabstractindextext{%
	\begin{abstract}
		Federated learning (FL) is a privacy-aware data mining strategy keeping the private data on the owners' machine and thereby confidential. The clients compute local models and send them to an aggregator which computes a global model. In hybrid FL, the local parameters are additionally masked using secure aggregation, such that only the global aggregated statistics become available in clear text, not the client specific updates. Federated QR decomposition has not been studied extensively in the context of cross-silo federated learning. In this article, we investigate the suitability of three QR decomposition algorithms for cross-silo FL and suggest a privacy-aware QR decomposition scheme based on the Gram-Schmidt algorithm which does not blatantly leak raw data. We apply the algorithm to compute linear regression in a federated manner. 
	\end{abstract}

}

	\maketitle
	
	\IEEEdisplaynontitleabstractindextext
	\IEEEpeerreviewmaketitle
	
    	\section{Introduction}
	Federated learning has risen in popularity following the seminal article by McMahan \textit{et al} \cite{McMahan2017}, and possibly accelerated by a search for new privacy preserving data analysis techniques following the introduction of the GDPR in Europe. Federated learning is a data analysis paradigm, where the data stays on the data owners' machine and only aggregated parameters are exchanged with the other participants or a central aggregator. There are two main versions of federated learning, cross-silo federated learning and cross-device federated learning. Cross-device FL connects many devices with relatively low computational power, such as mobile phones or sensors in a learning process. The devices have access to limited data, for example for one user. Cross-silo federated learning, the learning paradigm adopted in this article, joins multiple data silos containing records for a larger group of participants together \cite{Kairouz2021}. The federated setting adopted in this article is a type of hybrid federated learning which relies on secure parameter aggregation (SMPC). This means the computations at the client sides are done on clear text, but the aggregation is performed using secure multiparty computation. Therefore, only the aggregated parameters become known to the aggregator, not the individual client's updates. The participants are honest-but-curious, following the protocol, but trying to infer as much information as possible from the updates they receive \cite{cramer2015secure}. Since we \enquote{only} use secure aggregation and allow the disclosure of intermediate and final results, the advantage is, that we can directly chain together different algorithms into pipelines. Modern data analysis workflows rarely only use a single tool, therefore the use of secure aggregation allows reasonable privacy guarantees, without the need to develop new protocols for every workflow.

	We recently identified federated QR orthonormalisation as a contributor to a more privacy preserving principal component analysis (PCA) algorithm \cite{Hartebrodt2021}. Using federated QR orthonormalization for singular value decomposition allows the right, patient-associated singular vectors to remain private when using federated power iteration. QR decomposition is a versatile tool used for many more applications in linear algebra, including the solution of systems of linear equations \cite{linalgFord}.
	
	In centralized learning, the traditional machine learning setup, where all data is on a global server, three algorithms for QR factorization are available. They are based on Householder reflection, Givens rotation and the Gram-Schmidt procedure. The Householder algorithm is the most efficient for general applications, while Givens rotation is advantageous for sparse matrices and parallel computing architectures \cite{Sameh1971}. Gram-Schmidt orthonormalization is not used as much in practice due to numerical instabilities on special matrices \cite{linalgFord}. However, a stabilized version of the algorithm exists and privacy considerations may take precedence over numerical issues. Consequently, it is interesting to evaluate the algorithms with regards to their suitability for federated learning. The primary goal of this article is to evaluate the data disclosure of the three algorithms when deployed in a federated setting. In an earlier article, we introduced federated Gram-Schmidt orthonormalisation \cite{Hartebrodt2021}, but it does not return the full decomposition. Therefore, it needs to be extended to return a full QR factorization. The other algorithms have not been explicitly introduced for cross-silo FL, therefore in this article we develop prototypes of their federated versions. We show that Householder reflection and Givens rotation have properties that render them unsuitable for federated computation, even when secure aggregation is used.
	
	There are many federated algorithms that can be used for PCA, including a QR based scheme which has been introduced by \cite{Bai2005} and which we extend using the orthonormalization scheme developed in this article. We consider the privacy of this scheme, when using federated QR orthonormalization. Notably, the question we want to answer is whether the introduction of the federated QR scheme increases the privacy of the algorithm. We find, that upper triangular matrices are vulnerable to data leakage even when applying Gram-Schmidt orthogonalisation.

    Furthermore, to highlight the use for the federated QR procedure in other applications, we apply federated QR decomposition to compute linear regression. This could be used as an alternative solver for the federated linear regression computation suggested for instance in \cite{Nasirigerdeh2022}. Our experiments demonstrate the same accuracy of the federated linear regression as standard standalone tools.
	
	To summarize, our contributions are the following:
	\begin{itemize}
		\item We develop prototypes of federated Householder reflection and Givens rotation and extend a previous algorithm to return the full QR decomposition.
		\item We analyze the Householder, Givens and Gram-Schmidt algorithms for QR decomposition with respect to their privacy when using hybrid federated learning with secure additive aggregation.
		\item We conclude that both federated Householder reflection and Givens rotation introduce critical data leaks even when using secure additive aggregation.
		\item We investigate a special application case of federated Gram-Schmidt orthogonalisation on upper triangular matrices which may expose the input data.
		\item We provide a realistic use case of federated QR decomposition, linear regression.
	\end{itemize}
	
	The remainder of this manuscript is organized as follows: in \Cref{sec:qr:preliminaries}, the preliminaries, including the three centralized QR algorithms are introduced. \Cref{sec:qr:related} introduces related work. Based on the centralised descriptions, in \Cref{sec:qr:federated-qr}, we develop federated QR schemes for all algorithms, and a more detailed description for the most suitable QR algorithm, the federated Gram-Schmidt procedure (\Cref{sec:qr:fed-gs}). In \Cref{sec:qr:pca}, we evaluate the privacy of a QR based PCA scheme. \Cref{sec:fed-reg} describes how to compute federated linear regression based on QR decomposition. We provide empirical results for our analyses in \Cref{sec:experiments}. Lastly, the results are briefly discussed in \Cref{sec:qr:discussion}. \Cref{sec:qr:conclusion} concludes the work.
	
		\section{Preliminaries}
	\label{sec:qr:preliminaries}
	\subsection{Data model and architecture}
	In this manuscript we assume matrix $\A \in \mathbb{R}^{n \times m}$ to be partitioned into a set of $s \in [S]$ partial data sets such that $\A^s \in \mathbb{R}^{n^s \times m}$. $[S]$ denotes the set of clients joining the learning system. This partitioning is referred to as horizontal. We assume all participants have a share of the data, and the ordering of the rows is known and fixed. We describe our algorithm using a star-like architecture. We expect 
	the parameters to be masked using additive secure aggregation (\cf \Cref{sec:qr:preliminaries:secure-aggregation}), therefore we assume that peer-to-peer communication is possible via secure channels, regardless of the underlying architecture. This implies that our algorithms could be run on a fully decentralized architecture. The main reason for the choice of an aggregator-based architecture is the reduction in overall communication, because without SMPC the clients do not have to transmit the intermediate parameters to all their peers, only to the aggregator. 
	
	\subsection{Notation}
	\label{sec:qr:preliminaries:notation}
	Vectors and matrices are denoted in boldface, scalars in normal font. Matrices are noted in upper case letters and consist of column vectors which are noted in lower case letters. For instance, the matrix $\A^{n\times m}$ consists of $m$ column vectors $\a_i$ where $i$ is the index of the column. Sometimes we refer to columns and rows of a matrix as $\A_{\bullet, i}$ and $\A_{i, \bullet}$ respectively. \Cref{tab:qr:notation} contains an overview over the most frequently used variables in this work.
	
	\begin{table}[ht]
		\centering
		\caption{Notation table.}\label{tab:qr:notation}
		\begin{tabular}{ll}
			\toprule
			Syntax & Semantics \\
			\midrule
			$[N]\subset\mathbb{N}$ & index set $[N]=\{i\in\mathbb{N}\mid1\leq i\leq N\}$ \\
			$S\in\mathbb{N}$ & number of sites \\
			$m\in\mathbb{N}$ & number of features \\
			$n\in\mathbb{N}$ & total number of samples \\
			$n^s\in\mathbb{N}$ & number of samples at site $s\in[S]$ \\
			$\A\in\mathbb{R}^{n \times m}$ & complete data matrix \\
			$\A^s\in\mathbb{R}^{n^s \times m}$ & subset of data available at site $s\in[S]$\\
			$\U\in\mathbb{R}^{n\times k}$ & an orthogonal matrix with $\thespan(\U)=\thespan(\A)$ \\
			$\U^s\in\mathbb{R}^{n^s\times k}$ & the sub matrix of $\U$ available at sites $s \in [s]$ \\
			$\Q\in\mathbb{R}^{n\times k}$ & an orthonormal matrix with $\thespan(\Q)=\thespan(\A)$ \\
			$\Q^s\in\mathbb{R}^{n^s\times k}$ & the sub matrix of $\Q$ available at sites $s \in [s]$ \\
			$\R\in\mathbb{R}^{k\times k}$ & an upper triangular matrix\\
			\bottomrule
		\end{tabular}
	\end{table}
	
	\subsection{Secure aggregation}
	\label{sec:qr:preliminaries:secure-aggregation}
	The secure aggregation scheme used in this work relies on the additive aggregation protocol used by \cite{cramer2015secure}. It assumes honest-but-curious participants, \ie all clients perform the computations following the protocol but try to infer as much information as possible from the exchanged parameters \cite{Snyder2014YaoS}. All $s$ clients create $i=S$ random shares $x_{s,i}$ of their secret value $x_s$ such that $\sum_{i}^{S} x_{s,i} \mod p = x_s$, where $p$ is a large prime known to all participants. One can think of the \enquote{s} in $x_{s,i}$ as the source of the share and \enquote{i} the destination. All clients send the respective shares $x_{s,i}$ to the respective recipients $i$ where the sum of the shares is computed as $\sum_{s}^{S} x_{s,i} = x_i$. None of the shares disclose any information on the original values. Lastly, the clients announce their aggregated secret share $x_i$, such that the global sum $x = \sum_{i}^{S} x_i \mod p$  of all private shares can be formed. This scheme is suitable for a cross-silo federated learning systems with reliable clients (\ie they do not randomly drop out) and relatively few participants. Other secure aggregation schemes, such as Shamir's protocol, which are more fault tolerant to client dropout \cite{cramer2015secure} could be used instead without conceptual change of the algorithm. 
	
	\subsection{Centralized QR decomposition}
	\label{fed:qr:preliminaries:centralized-qr}
	The QR decomposition is the factorization of a square matrix into a square orthonormal matrix $\Q$ and an upper triangular matrix $\R$.
	
	\begin{equation}
	\A = \Q\R
	\end{equation}
	
	It exists also for non-square matrices (reduced QR decomposition) which is significantly more memory efficient if $n>m$. Three popular schemes exist for the computation of the decomposition, the Householder, Givens and Gram-Schmidt algorithms. In centralized systems, the Householder algorithm and Givens rotation are more popular, because they do not suffer numerical instability as the canonical version of Gram-Schmidt orthonormalization. Generally, Householder reflection is more efficient, and preferred unless the matrices are sparse or parallel compute architecture can be used \cite{Sameh1971}. See \cite{linalgFord} for more details on the algorithms.
	
	\subsubsection{Householder transformation}
	\label{sec:qr:prelim:householder}
	The Householder reflection proceeds column-wise, setting all the elements below the diagonal to 0 using a Householder reflector. Therefore, it requires $m-1$ Householder reflections to form an upper triangular matrix $\R$ starting from a matrix $\A \in \mathbb{R}^{n \times m}$. A Householder reflector is defined as
	
	\begin{equation}
	\Q^u = I-\frac{2\u\u^\top}{\u^\top \u}, \u \ne 0
	\end{equation}
	
	For each column vector $\a_i$ in matrix $\A$ the Householder reflection is computed using the following steps. First, $\a$ is normalized as $\a_i$ = $\frac{\a_i}{||\a_i||_\infty}$ to avoid numerical overflow. Then the vector $\u$, \ie the vector required for the construction of the Householder reflector, is computed by $\u_i = \a_i \pm ||\a_i||_2\cdot\e$, where $\e = \begin{bmatrix} 1 &0 &\cdot & 0\end{bmatrix}^\top \in \mathbb{R}^{m\times 1}$ denotes a vector of length $m$ containing a 1 in the first position and 0 otherwise. For ease of notation, the scaling factor $ \frac{2}{\u^\top\u}$ is denoted $\beta$. The Householder reflection is computed implicitly to increase the computational performance. The resulting matrix $\Q^u\A = \A - \beta\u\u^T\A$ contains 0 in the column corresponding to vector $\a_i$. \Cref{alg:qr:householder} summarizes the a single Householder reflection. In the Householder QR algorithm this operation is performed for all vectors $\a_i$ of $\A$, transforming in each step only the sub matrix, which is not yet upper triangular by choosing the remaining reflection matrix as the identity matrix $I$.  The full description of the Householder algorithm is shown in \Cref{alg:qr:householder-algorithm}.

	\begin{algorithm}[ht]
	\small
		\caption{Householder reflection}
		\label{alg:qr:householder}
		
		\KwInput{Data matrix $\A^s\in\mathbb{R}^{m \times m}$}		
		$\bar{\a} = \frac{\A_{\bullet, i}}{||\A_{\bullet, i}||_\infty}$\;
		$\u = \bar{\a} + sgn(\A_{1, i}) \cdot ||\bar{\a}||_2\cdot\e$\;
		$\Q_i\A = \A - \beta\u\u^T\A$\;
		\Return$ \A, \u$\;
	\end{algorithm}
	
		\begin{algorithm}
		\small
		\caption{Householder algorithm}
		\label{alg:qr:householder-algorithm}
		$\R \gets \A$
		$\Q \gets \textbf{I}$\;
		\KwInput{Data matrix $\A^s\in\mathbb{R}^{m \times m}$}	
		\For{$i \in m-1$}{
			\tcp{Update relevant entries of \R}
			$\R_{i:m, i:m}, \u \gets \householder(\R_{i:m, i:m})$\;
			\tcp{Implicit multiplication of Q}
			$\Q_{1:m, i:m} \gets \Q_{1:m, i:m} -(\frac{2}{||\u||_2^2})\Q_{1:m, i:m} (\u\u^\top)$\;
		}	
		\Return $\Q, \R$\;
	\end{algorithm}
	
	\subsubsection{Givens rotation} 
	\label{sec:qr:prelim:givens}
	Givens QR algorithm sequentially sets subdiagonal elements of the matrix $\A \in \mathbb{R}^{m \times n}$ to 0 by multiplying the matrix with the corresponding \enquote{Givens matrix} \cite{linalgFord}. After $\frac{n\cdot (m-1)}{2}$ operations all elements below the diagonal are 0 resulting in an upper triangular matrix $\R$. Through careful choice of the parameters in the Givens matrices, their product results in an orthogonal matrix $\Q$ which is the desired result.
	
	A Givens matrix has the following form, where $i$ and $j$ are the indices of $c$ and $s$.
	\begin{equation}
	\label{eq:givens-matrix}
	J(i,j,c,s) = 
	\begin{bsmallmatrix}
	1  & \cdots & 0 & \cdots & 0& \cdots & 0 \\
	0  & \ddots &  &  & &  & \vdots \\
	0   &  & c & \cdots & s&  & 0 \\
	\vdots  &  & \vdots & \ddots & \vdots&  & \vdots \\
	0 &   & -s & \cdots & c&  & 0 \\
	\vdots  &  &  &  & & \ddots & 0 \\
	0 &  \cdots & 0 & \cdots & 0& 0 & 1 \\
	\end{bsmallmatrix}
	\end{equation}
	$J(i,j,c,s)$ is orthogonal if $c^2+s^2=1$.
	
	Let $\A$ be the matrix of interest and $i$ and $j$ with $i<j$ indices of the element to be set to 0. Then one can set 
	\begin{equation}
	c = \frac{x_{i,i}}{\sqrt{x_{i,i}^2+x_{j,i}^2}}
	\end{equation}
	and 
	\begin{equation}
	s = \frac{x_{i,j}}{\sqrt{x_{i,i}^2+x_{i,j}^2}}
	\end{equation}
	
	and compute the respective Givens matrix according to \Cref{eq:givens-matrix}.
	
	Then $\A' = J(i,j, c,s)\A$ contains a 0 at position $(i,j)$.The product of a Givens matrix with a general matrix can be computed efficiently, by updating only rows $i$ and $j$ of the matrix as 
	
	\begin{equation}
	\A_{i,\bullet} =  [ca_{i,1}+sa_{j,1},ca_{i,2}+sa_{j,2},\cdots, ca_{i,m}+sa_{j,m}]
	\label{eq:givens-1}
	\end{equation} 
	and 
	\begin{equation}
	\A_{j,\bullet} =  [ca_{j,1}+sa_{i,1},ca_{j,2}+sa_{i,2},\cdots, ca_{j,m}+sa_{i,m}]
		\label{eq:givens-2}
	\end{equation} 
	
	The full QR decomposition in a centralized setting is summarized in \Cref{alg:qr:givens}.
	
	\begin{algorithm}[ht]
	\small
		\caption{QR factorization using Givens rotation}
		\label{alg:qr:givens}
		
		\KwIn{Data matrix $\A^s\in\mathbb{R}^{n \times m}$}
		\ForEach{$i \in [1, ..., m-1]$}{
			\ForEach{$j \in [i+1, ...,m]$}{
				$[s,c] \gets \givens()$\;
				$\A = J(i,j,c,s)\A$\;
				$\Q = J(i,j,c,s)\Q$\;
			}
		}
		$\R = A$\;
		$\Q  = \Q^\top$
		\Return $\Q,\R$
		
	\end{algorithm}
	
	\subsubsection{Gram-Schmidt orthonormalization}
	\label{sec:qr:prelim:cent-gs}
	The Gram-Schmidt algorithm produces an orthonormal matrix $\Q = [\q_1 \dots \q_k]$ and an upper triangular matrix $\R=[\r_1\dots \r_k]$ \cite{Beezer}. With a matrix $\A=[\a_1\dots\a_k]\in\mathbb{R}^{n\times m}$ of $m$ linearly independent column vectors, the matrix  $\U=[\u_1\dots\u_k]\in\mathbb{R}^{n\times m}$ of orthogonal column vectors is computed, such that it has the same span as $\A$.  Let $r_{i,j}=\u_j^\top\a_i/n_j$ and $n_j=\u_j^\top\u_j$ then
	\begin{equation}
	\u_i = \begin{cases}
	\a_i & \text{if }i=1\\
	\a_i - \sum_{j=1}^{i-1}  r_{i,j} \cdot \u_j & \text{if }i\in[k]\setminus\{1\}
	\end{cases}\text{,}
	\label{eq:qr-proj}
	\end{equation}
	
	\begin{equation}
	\q_{i} = \frac{\u_j}{||\u_j||} 
	\label{eq:qr:q}
	\end{equation}
	
	\begin{equation}
	\r_{j,i} = \begin{cases} 
	\q_j  \cdot \a_i & \text{if $j \leq i$}\\
	0 & \text{if $ j > i$}\\
	\end{cases}
	\label{eq:qr:r}
	\end{equation}
	
\subsection{Centralized Singular Value Decomposition}
Singular value decomposition (SVD) is a matrix decomposition frequently used in data mining applications. A matrix $\A$ is decomposed into two orthonormal matrices of singular vectors $\U$ and $\V$ and a diagonal matrix $\Sigma$ containing the singular values in non-increasing order $\A = \U\Sigma\V^{\top}$ \cite{linalgFord}. In the federated domain, SVD has been studied extensively, and multiple algorithms exist (\eg\cite{Bai2005, Hartebrodt2021, chai2019maskingsvd}). Given the vertically distributed matrix $\A^s \in \mathbb{R}^{m \times n^s}$ with dimension $m\times n^s$ at sites $s$ the federated singular value decomposition is defined as
	\begin{equation}
	\A^s  = \U\Sigma\V^{s\top}
	\end{equation}
	where $\U$ is the full left singular vector and $\V^s$ are the partial right singular vectors. The right singular vectors should not be shared due to potential privacy breaches \cite{Hartebrodt2021}.
	
	\subsection{Solution of systems of linear equations}
	\label{sec:qr:linalg-cent}
	In centralized computation, QR factorization can be used to compute the solution of systems of linear equations. Given a system $\A\x=\b$, one can compute $\A =\Q\R$. By setting $\Q\R\x = \b \Leftrightarrow \R\x = \Q^{-1}\b$ the system can be solved efficiently because due to the orthonormality of $\Q$, $\Q^{-1} = \Q^{\top}$ and $ \y = \Q^{-1}\b$ can be computed. This leaves to solve a system of the form $\R\x = \y$, which can be solved efficiently as $\R$ is an upper triangular matrix \cite{Beezer}. This can be used for instance for linear regression \cite{Nasirigerdeh2022}.

		\section{Related work}
	\label{sec:qr:related}
	Federated QR algorithms have been suggested mainly in the field of peer-to-peer networks relying on the PushSum algorithm and gossiping \cite{Sluciak2012, Sluciak2016, Strakov2012}. While these schemes can be implemented in a modern federated learning system, the assumptions governing FL make these algorithms unsuited. Notably, in cross-device FL, the client-to-client communication is assumed to be a bottleneck \cite{Kairouz2021} and client-aggregator communication is preferred. Secondly, cross-silo FL assumes more data and higher compute power at the nodes, so local computational constraints do not impact the computations as severely. In medical systems, practitioners might want to avoid approximation errors at the cost of higher compute time \cite{Yoo2021}. In distributed memory contexts, diverse schemes have been proposed to efficiently and quickly compute the QR decomposition (e.g. \cite{Song2010}). In these systems, there is usually one owner of the data, so we are unaware of privacy analyses in this context. In the outsourced, encrypted domain, the work of \cite{Zhang2019, Luo2017} still suggest very high expected execution times due to the encryption or masking overhead.
		\section{Federated QR Decomposition}
	\label{sec:qr:federated-qr}
	\begin{figure}[hb]
		\centering
			\begin{tikzpicture}[node distance=0.6cm]
	
	\node (A-vert) {\tikz\draw[step=0.25] (0,0) grid (1.5,-1);};
	
	\node[above = of A-vert, anchor=north, yshift=0.55cm] (A-vert-2) {\tikz\draw[step=0.25] (0,0) grid (1.5,1);};
	\node[fill=white,font=\Large,inner sep=1pt] (A-vert-label) at (A-vert-2) {\A};
	\node[above = of A-vert-2, anchor=north, yshift=0.05cm] (A-vert-3) {\tikz\draw[step=0.25] (0,0) grid (1.5,0.5);};
	\node[anchor=south,font=\small,yshift=-.1cm] at (A-vert-3.north) {$d$};
	\node[anchor=east,font=\small,yshift=-.1cm, xshift=0.2cm] at (A-vert.west) {$n_1$};
	\node[anchor=east,font=\small,yshift=-.1cm, xshift=0.2cm] at (A-vert-2.west) {$n_2$};
	\node[anchor=east,font=\small,yshift=-.1cm, xshift=0.2cm] at (A-vert-3.west) {$n_3$};
	
	\node[right = of A-vert, xshift=0.25cm] (Q-vert) {\tikz\draw[step=0.25] (0,0) grid (1.5,-1);};
	
	\node[above = of Q-vert, anchor=north, yshift=0.55cm] (Q-vert-2) {\tikz\draw[step=0.25] (0,0) grid (1.5,1);};
	\node[fill=white,font=\Large,inner sep=1pt] (Q-ckassic-label) at (Q-vert-2) {\Q};
	\node[above = of Q-vert-2, anchor=north, yshift=0.05cm] (Q-vert-3) {\tikz\draw[step=0.25] (0,0) grid (1.5,0.5);};
	\node[anchor=south,font=\small,yshift=-.1cm] at (Q-vert-3.north) {$d$};
	\node[anchor=east,font=\small,yshift=-.1cm, xshift=0.2cm] at (Q-vert.west) {$n_1$};
	\node[anchor=east,font=\small,yshift=-.1cm, xshift=0.2cm] at (Q-vert-2.west) {$n_2$};
	\node[anchor=east,font=\small,yshift=-.1cm, xshift=0.2cm] at (Q-vert-3.west) {$n_3$};


	\node[right = of Q-vert-3, yshift=-0.5cm, xshift=-0.5cm] (R-vert) {\tikz\draw[step=0.25] (0,0) grid (1.5,-1.5);};
	\node[fill=white,font=\Large,inner sep=1pt] (R-label) at (R-vert) {\R};
	\node[anchor=south,font=\small,yshift=-.1cm] at (R-vert.north) {$d$};
	
	\end{tikzpicture}
	
		\caption{Schematic QR decomposition with 3 participants. \A and \Q remain private. \R is known to all participants.}
		\label{fig:qr-schematic}
	\end{figure}
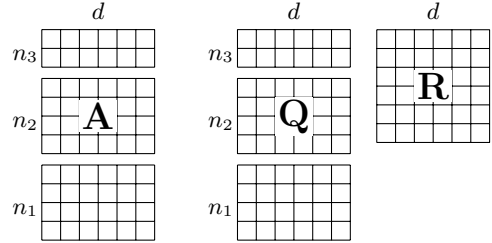
	
	In this section, we describe and analyse approaches to federate QR factorization. To our knowledge, there exist no descriptions of federated versions of the  Householder reflection or Givens rotation-based algorithms. Therefore, we first provide descriptions of the federated algorithms and demonstrate that they are not suitable for the chosen federated setting. Lastly, we describe the extended Gram-Schmidt algorithm which also returns the upper triangular matrix $\R$. Recall that we assume the data $\A \in \mathbf{R}^{n \times m}$ to be partitioned row-wise into chunks $\A^s \in \mathbf{R}^{n^s \times m}$. The goal of federated QR decomposition is to compute $\Q^s$ and $\R$ such that $\A^s$ and $\Q^s$ stay confidential, meaning the raw data does not leave site $s$ and $\Q$ can only be computed at $s$. $\R$ is common to all sites. 
	
	\subsection{Federated Householder algorithm}
	\label{sec:qr:fed-householder}
	We describe a straightforward algorithm for a federated Householder reflector. This subroutine could be used to compute the full QR decomposition in a federated manner. Let $t_s$ be the row index set of $\A^s$ at site $s$.

	\newcommand\tikzmk[1]{%
		\tikz[remember picture,overlay]\node[inner sep=2pt] (#1) {};}
	\newcommand\boxit[2][]{\tikz[remember picture,overlay]{\node[yshift=4pt, xshift=4pt, fill=#1,opacity=.3,fit={(A)($(B)+(#2\linewidth,.8\baselineskip)$)}] {};}\ignorespaces}

	\setlength{\fboxsep}{1pt}
	\colorlet{royal}{gray!70}

	\begin{algorithm}
	\small
		\caption{Federated Householder reflection \textcolor{royal!200}{Client-side computations are marked in gray.}}
		\label{alg:qr:federated-householder}
		
		\KwInput{Data matrices $\A^s \in \mathbb{R}^{n^s \times m}$ at sites $s\in[S]$.}
		\tcp{$t_s$ is the index set for the rows of $\A^s$}
		\tcp{Compute the global max element $||\A_{\bullet, i}||_\infty$ }
		\tikzmk{A}
		$m^s \gets \toAgg(||\A_{t_s, i}||_\infty)$\; \label{alg:householder:infinity-local}
		\tikzmk{B}
		\boxit[royal]{0.94}
		$||\A_{\bullet, i}||_\infty = \max_{s \in [S]} m^s$\; \label{alg:householder:infinity-global}
		\tikzmk{A}
		\tcp{Run at all clients}
		\For{$s\in[S]$}{
		\label{alg:householder:a-bar}
		\tcp{Compute the local portion of $\u$}
		$\bar{\a}^s \gets \frac{\A_{t_s, i}}{||\A_{\bullet, i}||_\infty}$\; 
		\tcp{Compute the norm of $\bar{\a}$}
		$n^s \gets \bar{\a}^s\top \bar{\a}^s$\;}
		\tikzmk{B}
		\boxit[royal]{0.94}
		$||\bar{\a}|| \gets \sqrt{\sum_{s} n^s}$\;
		\tikzmk{A}
		$\u^s \gets \bar{\a^s} \pm ||\bar{\a}||_2\cdot\e$\; 
		\tikzmk{B}
		\boxit[royal]{0.94}
		\label{alg:householder:u-local}
		\tcp{Oracle step: stack $\u^s$ to form $\u\u^\top$}
		$\u \gets \stackVert([\u_1,\cdots \u^s])$ \;\label{alg:householder:u-global}
		\tikzmk{A}
		\tcp{Update A}
		\For{$s\in[S]$}{
		$\H^s_A = \beta\u\u^T\A^s$\$\;
		}
		\tikzmk{B}\boxit[royal]{0.94}
        $\H_A = \sum_S \H^s_A$\;
        \tikzmk{A}
        \For{$s\in[S]$}{
		$\A^s= \Q^s_i\A^s = \A^s -\H_A$\;
		}\label{alg:householder:reflection}
		\tikzmk{B}\boxit[royal]{0.94}
		\Return $\A^s$, $\u$\;
	\end{algorithm}

	In \Cref{alg:qr:federated-householder} we describe a federated Householder reflector. The algorithm proceeds column wise. Initially, the global infinity norm $||\A_{\bullet,i}||_\infty$ is computed as the max over all local infinity norms (line 2) and the norm of the scaled vector is computed (line 4 to 7). Then, the clients locally compute $\u$ (line 8). In order to compute the Householder reflector, the clients send their partial vectors $\u^s$ to the aggregator (line 10). We call this step an \enquote{oracle step} to indicate that under the chosen secure computation paradigm, the aggregation itself cannot be performed privately. The reflection matrix $\H_A$ needs to be computed collaboratively by adding up the local shares (line 14 to 16).  Finally, at the clients, the reflection is performed (line 16). 
	
	Ad-hoc, this naive federated implementation of the procedure would take four communication rounds per column vector, one for the computation of the maximal element, one for the computation of the norm, and two for the computation of the reflector $\H_A$ and the reflection.
	
	In the federated setting, the computation of the Householder reflector itself is immediately problematic regarding the confidentiality of the data. Recall that algorithm relies on the computation of the outer product of $\u$ which is a direct transformation of the original column vectors of $\A$. In step 6, we call this operation an oracle step because it cannot be performed using the SMPC scheme we choose. Furthermore, even if secure multiplication is used, this \enquote{summary statistics} constitutes a privacy breach because the diagonal of $\u\u^\top$ contains the squared entries of $\u$. If $\u$, and $||\a||_\infty$ are known, then the original vector $\A_{\bullet, i}$ can be reconstructed. (We still assume, that all aggregate statistics are known, but we assume that the 'oracle step' can be computed using secure multiplication.)
	
	\begin{proposition}
	\label{prop:householder}
	Assuming that only aggregate statistics, excluding $\u$, but including the Householder projector $\H_A = \beta\u\u^\top$ become known in clear text to any of the participants, they can reconstruct the entire input data based on the summary statistics they know.
	\end{proposition}
	\begin{proof}
	Recall that during the computation of the reflector, $m$ and $n$ the maximal element and norm of $\A_{\bullet,i}$; $\beta$, the scaling factor and $\u\u^\top$ become known to the participants. The diagonal of $\u\u^\top$ contains the squared elements of $\u$, which can hence be computed up to the sign. Using the fact, that every participant can compute their share $\u^s$ of $\u$, it is possible to infer the sign of $\u$ for all participants using the off diagonal entries of $\u\u^\top$. For two sites $s$, and $s'$, with index set $t_{s'}$ denoting the entries belonging to $s'$ in $\u\u^\top$, we can compute the sign of $\u^{s'}$ at site $s$: $sgn(\u^{s'}) = sgn(\u^s_1) \cdot sgn(\u\u^\top_{t_{s'}, 1})$. Once the sign of $\u$ is known, the linear transformations can be reversed and $\A_{\bullet,i}$ becomes known.
	\end{proof}

	Therefore, it is not straightforward to privately compute the Householder transform using hybrid federated learning with secure aggregation. Knowledge of the procedures allows the reverse engineering of the data. This can potentially be prevented by performing the entire computation under homomorphic encryption, or an SMPC scheme which allows the evaluation of arbitrarily complex circuits. When using SMPC, it would not be sufficient to compute the outer product securely, the intermediate parameters cannot become known to any of the computing parties. Based on the incompatibility of the Householder reflection with secure aggregation, we exclude federated Householder reflection from further considerations.
	
	\subsection{Federated Givens rotation}
	\label{sec:qr:fed-givens}
	In this section, we describe a direct translation of a Givens rotation to a federated setting. Again, we only describe the relevant subroutine which would allow the implementation of the complete QR decomposition, albeit inefficiently. Realistically, one would choose a parallelized version of the operator.
	
	\begin{algorithm}[ht]
		\caption{QR factorization using Givens rotation \textcolor{royal!200}{Client-side computations are marked in gray.}}
		\label{alg:qr:givens-federated}
		
		\KwInput{Data matrices $\A^s \in \mathbb{R}^{n^s \times m}$ at sites $s\in[S]$.}
	\tcc{Perform local precomputations, setting all possible elements to $0$, send all non-zero indices to the aggregator}
		
		\ForEach{$i \in [1, ..., m-1]$}{
			\ForEach{$j \in [i+1, ...,m]$}{
				\tcp{Compute $c$ and $s$, using values from two clients $k_1,k_2 \in [S]$}
				$\toClient(i, j)$\; \label{alg:givens:announce}
				\tikzmk{A}
				$x_{i,i} \gets \toAgg(a^{k_1}_{\a_{i,i}})$\; \label{alg:givens:xi}
				$x_{j,i} \gets \toAgg(a^{k_2}_{\a_{j,i}})$\; \label{alg:givens:xj}
				\tikzmk{B}
				\boxit[royal]{0.82}
				$c = \frac{x_{i,i}}{\sqrt{x_{i,i}^2+x_{j,i}^s}}$\; \label{alg:givens:c}
				$s = \frac{x_{i,j}}{\sqrt{x_{i,i}^2+x_{i,j}^s}}$\; \label{alg:givens:s}
				\tikzmk{A}
				\tcp{Exchange the relevant entries required for the rotation}
			    $\toAgg [sa_{j,1},sa_{j,2},\cdots, sa_{j,m}]$\;
                $\toAgg ([sa_{i,1},sa_{i,2},\cdots,sa_{i,m}])$\;
                \tcp{Perform rotation following \Cref{eq:givens-1} and \ref{eq:givens-2}}
				$\A^s = J(i,j,c,s)\A^s$\; \label{alg:givens:update-A}
				$\Q^s = J(i,j,c,s)\Q^s$\; \label{alg:givens:update-Q}
				\tikzmk{B}\boxit[royal]{0.82}}}
		$\R = \A$\;
		$\Q  = \Q^\top$\;
		\Return $\Q,\R$
		
	\end{algorithm}

	\Cref{alg:qr:givens-federated} summarizes the federated procedure described in the following. As precomputations, the clients perform Givens rotations to set all elements to 0 which only depend on their data. Then, the clients communicate all remaining non-zero indices below the diagonal to the aggregator. Setting an element to $0$ requires only two rows $i$, and $j$ to be manipulated. The clients associated with these rows are called $k_1$ and $k_2$. In the main loop, the aggregator announces the current $i$ and $j$ to the current clients $k_1$ and $k_2$ (\Cref{alg:givens:announce}). Client $k_1$ and $k_2$ compute and announce the Givens parameters $s$ and $c$ in collaboration with the aggregator (\Crefrange{alg:givens:c}{alg:givens:s}). This is an \enquote{oracle step}, as this implies the communication of $x_i$ and $x_j$, because the Givens parameters cannot trivially be to compute using secure addition. The aggregator announces $c$ and $s$ to $k_1$ and $k_2$ and the clients update $\R$ and $\Q$. The broadcast can be combined with the new index broadcast (\Cref{alg:givens:announce}) if applicable. \Crefrange{alg:givens:announce}{alg:givens:update-Q} are repeated until all elements below the diagonal are $0$.
	
	The naive implementation of this procedure would require in the order of  $N=\mathcal{O}(\frac{2\cdot n \cdot (m-1)}{2})$ transmission rounds. Each element, would required an index broadcast and a Givens parameter broadcast. The procedure can be parallelized to zero out $\frac{n}{2}$ elements per round \cite{Sameh1971}, reducing the communication complexity to $\mathcal{O}(n)$. Local precomputation would decrease the effective number of transmission costs. Furthermore, the index broadcast can most likely be done with fewer communications rounds.
	
	However, there is a critical privacy breach when using Givens rotations. Recall that we assume the data to partitioned into $s$ partitions $A^s \in \mathbb{R}^{n^s \times m}$. Assuming rows $i$ and $j$ are located in silo $S_1$ and $S_2$ respectively, the aggregator can compute the values $c$ and $s$ using $x_i$ and $x_j$  (\cf \Cref{eq:givens-1}, \Cref{eq:givens-2}). Even if $c$ and $s$ are computed using SMPC and P2P communication (so that the aggregator does not gain knowledge of the parameters), $x_i$ and $x_j$ can be reconstructed at the current clients $k_1$ and $k_2$. 
	
	\begin{proposition}
	\label{prop:givens}
	   A client $S$ can reconstruct an entire row of the sub matrix $\A^{S'}$ of another participant $S'$, if the Givens parameters as well as $\a^0_i$ and  $\a^1_i$, the rows of $\A^S$ before and after the update are available in clear text at client $S$.
	\end{proposition}
	
	\begin{proof}
	Let $a^0_i$ denote the $i$th row of $\A$ at client $k$ before the rotation and $a^1_i$ denote the $i$th row at client $k$ after the update. Given the Givens parameters $c$ and $s$, client $k$ can compute the $j$th row at client $k'$ as $a^0_j =[(a^1_{il}-c\cdot a^0_{il})/s]$ with $l$ the column index.
	\end{proof} 
	
	In order to prevent this breach, the whole algorithm would have to be performed under encryption, such that $S_1$ does not gain access to the intermediate matrices $A^{S'}$. These considerations render this algorithm unsuitable for hybrid federated learning with secure parameter aggregation. This leaves the Gram-Schmidt algorithm as the final possible algorithm.

	\subsection{Federated Gram-Schmidt Algorithm including the Computation of R}
	\label{sec:qr:fed-gs}
	Based on the algorithm described in \cite{Hartebrodt2021}, where we showed that the orthogonal matrix $\Q$ can be computed solely based on the exchange and aggregation of vector norms and co-norms, we extend the algorithm such that the $\R$ matrix can be computed simultaneously. This can be done without further communication steps in comparison to the previously presented method. The main modifications are that the orthonormal vectors in $\Q$ need to be computed right away in order to compute the inner product of $\q_l\a_{i-1}$ contained in the matrix $\R$ at position $l, i-1$. Note, that the procedure also requires the orthogonal vectors $\u_i$.

	Therefore, we develop a detailed description of a federated Gram-Schmidt orthonormalization procedure (see \Cref{alg:qr-federated-new}). First, the global vector norm $n_i$ of $\u_i$ is calculated by computing the local vector norms $n^s_i$ at the clients and aggregating them at the central server (\Crefrange{alg:gs:local-norms-start}{alg:gs:local-norms-end}). The main loop starts at index $i=2$ and proceeds in \num{4} stages. Let $\R$ be the upper triangular matrix completed up to vector $i \in d$. First, $\e_{i-1}^s$, is computed by dividing $\u_{i-1}$ through the global norm (\Crefrange{alg:gs:compute-q-1-start}{alg:gs:compute-q-1-end}). Then, the $i-1$st local column $\r^s_{l, i-1}$ of $\R$ is computed as the inner product of the partially normalized vector $\q_{i-1}^s$ and the partial data column $\a^s$ (\Crefrange{alg:gs:r-local-start}{alg:gs:r-local-end}). Then the local residuals $r^s_{ij}$ for vector $i$ \wrt to the previous $i-1$ vectors are computed (\Crefrange{alg:gs:residuals-start}{alg:gs:residuals-end}). In stage 2, the two parameters $\r^s_{l, i-1}$ and $r^s_{ij}$ are sent to the central server and aggregated via element-wise addition (\Crefrange{alg:gs:aggregated-start}{alg:gs:aggregated-end}) to form the global copy of $\R$ up until $i-1$. The global $\r_{l, i-1}$, and $r_{ij}$ are returned to the clients, where the orthogonal vector $\u^s_i$ is computed (\Crefrange{alg:gs:orthog-start}{alg:gs:orthog-end}). In the last stage, the norm of the current vector $\u_i$, $n_i$ is computed by summing up the local norms of $\u^s_i$ (\Cref{alg:gs:aggregate-vnorm}). The procedure is repeated for all $d$ vectors of $\A$. After exiting the main loop, the last column of \A is computed, and the partial orthonormal matrices $\Q^s$ and $\R$ are returned (\Crefrange{alg:gs:final-start}{alg:gs:final-end}). This procedure is equal to the centralized Gram-Schmidt algorithm because the vector inner products can be computed exactly in a federated fashion.

	\begin{algorithm}[ht!]
	\small
		\caption{Federated Gram-Schmidt. \textcolor{royal!200}{Client-side computations are marked in gray.}}
		\label{alg:qr-federated-new}
		\KwInput{Data matrices $\A^s \in \mathbb{R}^{n^s \times m}$ at sites $s\in[S]$.}
		\KwOutput{Partial matrices $\Q^s$ and full matrix $\R$ at sites $s\in[S]$ }
		\tikzmk{A}
		\tcp{Compute norm of first orthogonal vector.}
		\For{$s\in[S]$\label{alg:gs:local-norms-start}} 
		{$\u^s_1 \gets \a^s_1$\;  \label{alg:gs:utu-1}
			$n_1^s \gets {\u_1^s}^\top{\u_1^s}$\; \label{alg:gs:n-1}}
		
		$n_1 \gets \sum_{s=1}^S n_1^s$ \; \label{alg:gs:n-1-sum} \label{alg:gs:local-norms-end}
		
		\tcp{Orthogonalize all subsequent vectors.}
		\For{$i \in [d]\setminus\{1\}$}{
			\tikzmk{A}
			\tcp{For each client $s$}
			\For{$s \in [S]$}{\label{alg:gs:compute-q-1-start} 
				\tcp{Normalise to unit norm}
				$\q^s_{i-1} \gets \u^s_{i-1} /\sqrt{n_{i-1}}$ \;  \label{alg:gs:compute-q-1-end} 
				\tcp{Compute relevant entries for \R}
				\For{$l \in [i-1]$ \label{alg:gs:r-local-start}}
				{$\r^s_{l,i-1} \gets \q^{s\top}_{l} \a^s_{i-1}$\;\label{alg:gs:r-local-end}}

				\tcp{Compute client residuals for current vector.} \label{alg:gs:residuals-start}
					\For{$j \in[i-1]$}{
						$r_{ij}^s \gets {\u_j^s}^\top \a_i^s / n_j $\; \label{alg:gs:residuals-end}
					}
			}
			\tikzmk{B}
			\boxit[royal]{0.88}
			\tcp{Aggregate residuals}
			\For{$j\in [i-1]$}{
				\label{alg:gs:aggregated-start}
				$r_{ij} \gets \sum_{s=1}^Sr_{ij}^s$\; \label{alg:gs:agg-res}
				\tcp{Aggregate R}
				\For{$l \in [i]$}
				{$\r_{l,i-1} \gets \sum_{s=1}^S \a^s_{l,i-1}$\; \label{alg:gs:aggregated-end}
				}
			}	
			\tikzmk{A}
			\tcp{Orthogonalize vector and compute norm.}
			\For{$s \in [S]$}{
				\label{alg:gs:orthog-start}
				$\u_{i}^s \gets \a_{i}^s - \sum_{j =1}^{i-1} r_{ij} \cdot \u_{j}^s$\;
				$n_i^s \gets {\u_i^s}^\top \u_i^s$ \; \label{alg:gs:orthog-end}
			}
			\tikzmk{B}
			\boxit[royal]{0.88}	
			$n_i \gets \sum_{s=1}^S n_i^s$ \; \label{alg:gs:aggregate-vnorm}
		}
		\tikzmk{A}
		\For{$s\in[S]$ }{\tcp{Compute last column of R}
			\label{alg:gs:final-start}
			$\q^s_{d} \gets \u^s_{d} /\sqrt{n_{i-1}}$\; 
			\For{$l \in [k]$}
			{$\r_{ld} \gets \q^s_{d} \a^s_{l}$\;}
			$\Q^s = [\q^s_1 \cdots \q^s_d]$ \;
			\textbf{Return} $\Q^s, \R$\; \label{alg:gs:final-end}}
		\tikzmk{B}
		\boxit[royal]{0.94}
	\end{algorithm}
	
	\subsection{Privacy considerations}
	Recall that according to our privacy definition, private federated QR decomposition returns $\Q^s$ and $\R$ such that $\A^s$ and $\Q^s$ stay private, meaning the raw data does not leave site $s$ and $\Q^s$ can only be computed at $s$. $\R$ is common to all sites. 
	
	\begin{proposition}
	    At the end of federated Gram-Schmidt decomposition, the clients do not have access to more knowledge than their data matrices $\A^s$, the orthonormal partial matrices $\Q^s$, and the global matrices $\R$. 
	\end{proposition}
	\begin{proof}
	We consider the case, where we have no knowledge of the type of matrix (for instance, whether it is sparse, or triangular) to be orthonormalized and analyze the knowledge at the aggregator. Let $\A^s=\Q^s\R$. At the end of the algorithm, the following knowledge is available at the aggregator (We only show the global aggregates, assuming that they are aggregated using secure addition):
	
	\begin{itemize}
		\item $[n_1, \cdots ,n_d]$, the norms of $[\u_1, \cdots, \u_d]$
		\item $\R$ the upper triangular matrix 
		\begin{equation}
		\begin{bmatrix}
		\q_1\cdot \a_1 & \q_1\cdot \a_2 & \cdots & \q_1\cdot \a_d \\
		0 & \q_2\cdot \a_2 & \cdots & \q_2\cdot \a_d \\
		\vdots & 0 & \ddots & \vdots \\
		0 &0 & \cdots & \q_d\cdot \a_d \\
		\end{bmatrix}	
		\end{equation}
		\item the upper triangular matrix of residuals
		\begin{equation}
		\begin{bmatrix}
		\u_1\cdot \a_2 & \u_1\cdot \a_3 & \cdots & \u_1\cdot \a_d \\
		0 & \u_2\cdot \a_3 & \cdots & \u_2\cdot \a_d \\
		\vdots & 0 & \ddots & \vdots \\
		0 &0 & \cdots & \u_d\cdot \a_d \\
		\end{bmatrix}	
		\end{equation}
		\item In particular, we do not have access to the matrices $\U^s$, $\Q^s$ or $\A^s$.
	\end{itemize}
	
	Since $\q_i=\frac{\u_i}{n_i}$, the total information available amounts to the information encoded in the $\R$ matrix. We hence have only access to one factor of the decomposition which does not allow us to find a unique solution to $\A=\Q\R$. We specified our privacy goal as keeping the input matrices $\A^s$ and the orthogonal matrices $\Q^s$ private, therefore the presented algorithm is private as per our definition.
\end{proof}
		
It should be noted that $\R$ does disclose information on the data in form of the feature covariance matrix:
	
	\begin{equation}
	\A^\top\A = \R^\top\Q^\top\Q\R = \R^\top\R
	\label{eq:pca-qr}
	\end{equation}

	\section{Further privacy investigations}
\label{sec:qr:pca}
In this section, we apply federated QR factorization as a subroutine in federated PCA to reveal a privacy breach that can occur, if secure aggregation is not used, or if only 2 parties participate in the computation. The original algorithm  uses QR factorization as the aggregation step \cite{Bai2005}. This centralized procedure can be replaced by federated QR orthonormalization, presumably preventing the disclosure of the local summary statistics. The algorithm is mainly of academic interest, because more efficient schemes for PCA are available for star-like architectures. However, we will show that knowledge of the procedure allows an honest-but-curious participant to exactly reconstruct the other participants' input data. Our attack exploits the fact, that the input matrices are upper triangular and that we have full knowledge of the algorithm. 
	
\subsection{Algorithm}
The algorithm \cite{Bai2005} relies on sending a local $\R$ to the aggregator, where a secondary QR decomposition is performed (\Cref{fed-pca:full-r}). We suggest centering the data globally prior to the computation of the matrix (\Cref{fed-pca:center}). This implies subtracting the mean from each column and dividing by the standard deviation to obtain variables with a mean of 0 and a variance of 1. This also avoids having to account for inter site differences in mean later on. The next step is identical to the original: all the $\R$ matrices are computed at the clients (\Cref{fed-pca:compute-R}). The original algorithm recursively merges the $\R$ matrices at a processor to form the updated $\R'$ matrix until only one matrix remains. By computing the QR decomposition of all clients' $\R$ matrices at once using federated Gram-Schmidt decomposition, sending $\R$ can be avoided. The federated QR algorithm returns \R at all the clients, therefore the final SVD can be directly computed at the client. The clients can also compute the partial left eigenvectors as $U^s=A^sV$ (\Cref{fed-pca:Compute-Us}). 

\begin{algorithm}[ht]
		\caption{Federated PCA using QR factorization \cite{Bai2005}}
		\label{fed-pca:full-r}
		\KwInput{Data matrices $\A^s\in\mathbb{R}^{n^s \times m}$, \# eigenvectors $k$.}
		$\A^s \gets \fedNorm()$ \; \label{fed-pca:center}
		\tikzmk{A}
		\For{$s \in [S]$}
		{\tcp{Compute local R at all clients}
			$\Q^s,\R^s \gets \orthonormalize(\A^s)$\; \label{fed-pca:compute-R}}
		
		\tikzmk{B}
		\boxit[royal]{0.94}
		$[\Q^s], \R \gets \fedOrthonormalize([\R^1, \cdots \R^s ])$\; \label{fed-pca:ortho-R}
		
		\tikzmk{A}
		$\U, \Sigma, \V^\top =  \svd(\R)$\; \label{fed-pca:R-SVD}
		$\U^s \gets \A^s\V$\; \label{fed-pca:Compute-Us}
		\tikzmk{B}
		\boxit[royal]{0.94}
		\textbf{Return} $\U^{s,k}, \V^k$
	\end{algorithm}

\subsection{Privacy of federated Gram-Schmidt on upper triangular matrices}
\label{sec:qr:pca:privacy}
In the original algorithm, the communication of \R poses a problem: Let $\A^\top\A$ be the covariance matrix of the data. Using the fact that \Q is an orthonormal matrix, \R can be used to compute the local covariance matrices of the data and hence leaks information (\Cref{eq:pca-qr}). Therefore, this algorithm is no more private than sending the entire set of local eigenvectors to the next party. The advantage of \Cref{fed-pca:full-r} over its previous version is that it allows the computation of the global $\R$ without communicating the local $\R$ in clear text. The same can be achieved by using secure addition of the covariance matrices or computing the global $\R$ based on the data instead of $\R$. Nonetheless, we investigate this algorithm, because with close analysis it reveals a privacy breach if secure aggregation is not used or only two participants join. We show, that in this case the federated QR decomposition of upper triangular matrices is no more private than sending the upper triangular matrices themselves. The reason for this is the fact that the initial vector norm of the QR step is not technically an aggregate. We visualize the aggregation step in \Cref{fed-pca:full-r} in \Cref{eq:problem}, as it is the motivation for our investigation. To avoid ambiguity, we denote the resulting upper triangular matrix $\S$ with elements $s_{i,j}$. For the remainder of this section, we assume that secure aggregation is not used.
	
	\begin{proposition}
		Let $\R^* = \begin{bmatrix} \R^1 & \R^2 & \cdots & \R^s \end{bmatrix}^\top$ be a vertical stack of upper triangular matrices, of which we want to compute the QR decomposition as $\R^* = \Q\S$. Denote $\Q^s = [\u_1^s, \u_2^s, \cdots \u_d^s]$ the block wise orthogonal matrices at sites $s$. It is possible to reconstruct all $\begin{bmatrix} \R^1 & \R^2 & \cdots & \R^s \end{bmatrix}$ as well as all $\begin{bmatrix} \Q^1 & \Q^2 & \cdots & \Q^s \end{bmatrix}$ when applying the federated QR algorithm on $\R^*$, given one knows that $\R^s$ are upper triangular.
	\end{proposition}

	\begin{equation}
	\label{eq:problem}
	\R^* = 
	\begin{bmatrix} \R^1\\ \R^2\\ \vdots\\ \R^s \end{bmatrix} = 
	\begin{bmatrix}
	\begin{psmallmatrix} 
	r^1_{11} & r^1_{12} & \cdots & r^1_{1d} \\
	0 & r^1_{22} & \cdots & r^1_{2d} \\
	\vdots & 0 & \ddots & \vdots \\
	0 &0 & \cdots & r^1_{dd} \\
	\end{psmallmatrix}\\ 
	
	\begin{psmallmatrix} 
	r^2_{11} & r^2_{12} & \cdots & r^2_{1d} \\
	0 & r^2_{22} & \cdots & r^2_{2d} \\
	\vdots & 0 & \ddots & \vdots \\
	0 &0 & \cdots & r^2_{dd} \\
	\end{psmallmatrix}\\
	\vdots\\
	
	\begin{psmallmatrix} 
	r^s_{11} & r^s_{12} & \cdots & r^s_{1d} \\
	0 & r^s_{22} & \cdots & r^s_{2d} \\
	\vdots & 0 & \ddots & \vdots \\
	0 &0 & \cdots & r^s_{dd} \\
	\end{psmallmatrix}
	\end{bmatrix}=
	\begin{bmatrix} \Q^1\\ \Q^2\\ \vdots\\ \Q^s \end{bmatrix} \S
	\end{equation}

	\begin{proof}
		Let $\R^* = \begin{bmatrix} \R^1 & \R^2 & \cdots & \R^s \end{bmatrix}^\top$ be the matrix to be decomposed into $\Q$ and $\S$. Denote $\R^s$ and $\Q^s$ the partial matrices only available at site $s \in [S]$. Denote $[\u^s_1 \cdots \u^s_d]$ the partial orthogonal vectors at sites $s$. We show by induction on $i$ that  $\R^s$ and $\Q^s$ can be reconstructed at the aggregator based on the intermediate summary statistics exchanged during the execution of \Cref{alg:qr-federated-new}. Let $i=1$. In the first step of the algorithm (\Crefrange{alg:gs:utu-1}{alg:gs:n-1-sum}, \Cref{alg:qr-federated-new}), when computing $n_1 = \sum_{s=1}^S \r_1^{s\top}\r_1^s$ the clients disclose $(\r^s_{1,1})^2$ to the aggregator which can compute
		
		\begin{equation}
		\u_1 = [\sqrt{\r^1_{1,1}},0, \cdots,0, \sqrt{\r^2_{1,1}}, \cdots,0, \sqrt{\r^s_{1,1}}, 0, \cdots, 0 ]^\top. 
		\end{equation}
		
		Let now $i=2$ and $j=1$, the residuals $p_{2,1}^s \gets  \frac{{\u_1^s}^\top\r_2^s}{n_1} $ are computed and aggregated as $p_{2,1} = \sum_{s=1}^Sp_{2,1}^s$ (\Cref{alg:gs:residuals-end} and \Cref{alg:gs:agg-res}). $n_1$ and $\u_1$ are known. We can compute $r^s_{1,1} =\q^{s\top}_1s^s_{1,1}$ because $\q^s_1$ is orthonormal and only contains a single non-zero entry (\Cref{alg:gs:r-local-start}, ($s$ corresponds to $r$ in the algorithm description)). For the same reason, we can also compute $r^s_{1,2} = \frac{p^s_{2,1}\cdot n_1}{q_{1,1}}$ (\Cref{alg:gs:residuals-end}).
		
		Finally, we compute $n_2 \gets \sum_{s=1}^S n_2^s$, with  $n_2^s = {\u_2^s}^\top \u_2^s$ where $\u_{2}^s \gets \r_{2}^s - \sum_{j =1}^{i-1} p_{21} \cdot \u_{1}^s$. This can be simplified to $\u_{2}^s = 
		\begin{psmallmatrix}
		r_{12}^s - p_{21}\cdot  u_{11}^s\\
		r_{22}^s
		\end{psmallmatrix}$, 
		because only $u^s_{1,1}$ is non-zero. $r^s_{1,2}$, $p_{2,1}$ and $u^s_{1,1}$ are known, so $r^s_{2,2} = \sqrt{n^s_2-(r_{1,2}^s - p_{2,1}\cdot u_{1,i}^s)^2}$ can be computed, which in turn means $\u^s_2$ is known completely. At this point, $\u_1$, $\u_2$, $\r^s_1$ and $\r^s_2$ are known to the aggregator. \newline
		
		For the inductive step, we assume to have computed $\Q$ and $[\R^1 \cdots \R^s]^\top$ up to column $i-1$, we can compute column $i$. We set $j=i-1$. 
		
		The residuals $p_{ij}^s \gets {\u_j^s}^\top \r_i^s / n_j $ are computed, and aggregated as $p_{ij} = \sum_{s=1}^Sr_{ij}^s$. $n_j$, $\r_j$ and $\u_j$ are known for $j \in [i-1]$. We can compute $r^s_{ij}$ for $j \in [i-1]$ via successive variable substitution due to the fact that the $\R^s_i$ are upper triangular. 
		\begin{equation}
		\raggedleft{
			\begin{cases}
			r_{1,i} =   \frac{p_{i,1}\cdot n_1}{u_{1,1}}\\
			r_{2,i} = \frac{p_{i,2} \cdot n_2-u_{1,2} \cdot r_{1,i}}{u_{2,2}}\\
			\vdots\\
			r_{j-1,i} = \frac{p_{i,j}\cdot n_j-\sum_{n=0}^{i}u_{n-1} \cdot p_{i, n-1}}{u_{n-1,n-1}}\\
			\end{cases}
		}
		\end{equation}
		
		Finally, we compute $n_i \gets \sum_{s=1}^S n_i^s$, with  $n_i^s = {\u_i^s}^\top \u_i^s$ where $\u_{i}^s \gets \r_{i}^s - \sum_{j =1}^{i-1} p_{ij} \cdot \u_{j}^s$ which  can be rewritten as 
		\begin{equation}
		\u_{i}^s = 
		\begin{pmatrix}
		r_{1,i}^s - \sum_{j=1}^{i-1} p_{i,j} \cdot  u_{1,j}^s\\
		r_{2,i}^s - \sum_{j=1}^{i-1} p_{i,j} \cdot  u_{2,j}^s\\
		\vdots\\
		r_{j,i}^s - \sum_{j=1}^{i-1} p_{i,j} \cdot  u_{j,j}^s\\
		r_{i,i}^s
		\end{pmatrix}
		\end{equation}
		where $\r_{i,i}$ is the only unknown.
		$\r_{ji} = \sqrt{n^s_i- \sum_{j=1}^{i-1}(\r_{j,i}^s -  p_{i,j}\cdot  \u_{j,i}^s)^2}$, which in turn means $\u^s_j$ is complete, because $n_i$, $p_{i,j}$ and $\u_j^s$ as well as $\r_{j,i}$ are known. 
	\end{proof}
	
	When using general matrices, even with only two participants, and if SMPC is used, this attack is not possible, because the first vector norm summarizes more than one element. However, the previous application highlights that tracking the parameters during federated iterations could reveal more information on the input data than the participants intend, especially when the methods are fully traceable and do not involve randomized steps. In the case of sparse matrices, a partial column $\a^s_i$ which contains no entries, can be detected at the aggregator as the inner product in $\R$ would be \num{0}. The problem, with the methods presented here, is that the knowledge of algorithmic procedure and the absence of random elements in the algorithm allow us to backtrack more information than intended, given an 'attack angle'.

\section{Systems of linear equations}
\label{sec:fed-reg}
To showcase a more realistic use of our algorithm, we consider the application of federated orthonormalization for the solution of systems of linear equations. A popular use of QR decomposition is linear regression. For example, \texttt{R}'s \texttt{lm()} function uses the QR algorithm by default \cite{Rlang}. Here, we demonstrate how it is possible to solve a system of linear equations of the form $\A\x = \b$ with only one further round of communication, based on the QR decomposition. This technique can be used to replace the solver implemented for instance in \cite{Nasirigerdeh2022}. It does not require matrix inversion and is therefore more suitable for large scale matrices. Let $\A$ and $\b$ be partitioned into $\A^s$ and $\b^s$ respectively, and $\x$ the solution common to all sites. After the QR decomposition of the matrix $\A$, $\Q^s$ is known at the sites, and $\R$ is known at all sites and the aggregator. In order to compute $\x$, the clients have to send their vector inner product of $\y^s = \Q^\top\b^s$ to the aggregator which securely computes the global vector $\y = \sum_{s}^{S} \y^s$. The aggregator can directly compute $\x$ by successive variable substitution and share the result with the clients (see \Cref{sec:qr:linalg-cent}). With one additional step, one can also compute p-values and $r^2$ statistics. The clients compute the sum of the squared residuals as $rss_s = \sum (\A^s\x - \b)^2$ and the sum of the squared fitted values $mss_s = \sum \A^s\b$ and send them to the aggregator, which computes the global sums $rss=\sum_s^S rss_s$, and $mss=\sum_s^S mss_s$.
For the p-value, the variance is computes as $\sigma = \frac{rss}{n-m-1}(\R^\top\R)^{-1}$, and standard error as $SE= \sqrt{\sigma}$. Here, we exploit the fact that the covariance matrix can be expressed using $\R$ (\Cref{eq:pca-qr}). The T-statistic used to determine the p-value can be computed as $T=\frac{\x}{SE}$. For more details see \cite{Nasirigerdeh2022} where a more detailed description of the p-value calculation is provided. $r^2$ can be computed as follows: $r^2 =  \frac{mss}{mss+rss}$.
	
	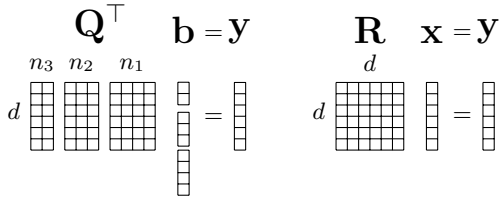
\begin{figure}[h]
		\centering
		\begin{tikzpicture}[node distance=0.35cm]

\node (Q-vert-3) {\tikz\draw[step=0.15] (0,0) grid (0.3,-0.9);};
\node[right = of Q-vert-3, anchor=west, xshift=-0.45cm] (Q-vert-2) {\tikz\draw[step=0.15] (0,0) grid(0.45, -0.9);};
\node [right = of Q-vert-2,xshift=-0.45cm] (Q-vert) {\tikz\draw[step=0.15] (0,0) grid (0.6, -0.9);};

\node[above = of Q-vert-2, fill=white,font=\Large,inner sep=1pt, xshift=0.25cm] (Q-classic-label) {$\Q^\top$};

\node[anchor=south,font=\small,yshift=-.15cm, xshift=-0.1cm] at (Q-vert-3.west) {$d$};
\node[anchor=south,font=\small,yshift=-.1cm] at (Q-vert.north) {$n_1$};
\node[anchor=south,font=\small,yshift=-.1cm] at (Q-vert-2.north) {$n_2$};
\node[anchor=south,font=\small,yshift=-.1cm] at (Q-vert-3.north) {$n_3$};

\node[right = of Q-vert, xshift=-0.3cm, yshift=-0.75cm] (y-vert) {\tikz\draw[step=0.15] (0,0) grid (0.15,-0.6);};

\node[above = of y-vert, anchor=north, yshift=0.15cm] (y-vert-2) {\tikz\draw[step=0.15] (0,0) grid (0.15,0.45);};
\node[above = of y-vert-2, anchor=north, yshift=0.05cm] (y-vert-3) {\tikz\draw[step=0.15] (0,0) grid (0.15,0.3);};
\node[above =of y-vert-3, fill=white,font=\Large,inner sep=1pt] (y-label) {\b};
\node[right =of y-vert-2, anchor=east,font=\small,xshift=0.5cm, yshift=0.15cm] (eq1) at (y-vert-2.west) {$=$};
\node[above =of eq1, anchor=south,font=\small, yshift=0.4cm] at (eq1.north) {$=$};
\node[right = of y-vert-3, yshift=-0.3cm] (x-vert) {\tikz\draw[step=0.15] (0,0) grid (0.15,-0.9);};
\node[above =of x-vert, fill=white,font=\Large,inner sep=1pt] (y-label) {\y};

\node[right = of x-vert, yshift=-0.3cm, xshift=0.6cm, yshift=0.3cm] (R-vert) {\tikz\draw[step=0.15] (0,0) grid (0.9,-0.9);};
\node[above = of R-vert, fill=white,font=\Large,inner sep=1pt] (R-label) {\R};
\node[anchor=south,font=\small,yshift=-.1cm] at (R-vert.north) {$d$};
\node[anchor=south,font=\small,yshift=-.15cm, xshift=-0.1cm] at (R-vert.west) {$d$};

\node[right = of R-vert, xshift=-0.3cm] (x-vert-2) {\tikz\draw[step=0.15] (0,0) grid (0.15,-0.9);};
\node[right =of x-vert-2, anchor=east,font=\small,xshift=0.5cm] (eq2) at (x-vert-2.west) {$=$};	
\node[above =of eq2, anchor=south,font=\small, yshift=0.4cm] at (eq2.north) {$=$};

\node[above =of x-vert-2, fill=white,font=\Large,inner sep=1pt] (y-label-2) {\x};

	\node[right = of x-vert-2] (x-vert-3) {\tikz\draw[step=0.15] (0,0) grid (0.15,-0.9);};

\node[above =of x-vert-3, fill=white,font=\Large,inner sep=1pt] (y-label-3) {\y};

\end{tikzpicture}
		\caption{Schematic solution of a system of linear equations based on federated QR factorization of matrix $\A$.}
		\label{fig:fed-lin-sys}
	\end{figure}
	
	\section{Experiments}
\label{sec:experiments}
\subsection{Data reconstruction}
In order to show that Householder reflection and Givens rotation indeed lead to confidentiality losses when the intermediate parameters become known, we implemented the federated prototypes and logged the parameters that are known to all participants. We generate random Gaussian matrices of dimension $5000\times 10$ and execute the federated prototypes. We then reconstruct the input data based on these aggregate statistics.

For the reconstruction of the input matrix based on the parameters disclosed during a Householder transformation, we log the reflection matrix $\u\u^\top$, $\beta$, $sgn(\A_{1,1})$, $m$, and $n$ and the first element of $\u$. (We assume without loss of generality that the first client is the attacker). We then apply the reconstruction described in \Cref{prop:householder}. We reconstruct the data with an error of $2.21^{-15}$ averaged over $10$ iterations. The error is calculated element-wise difference of the input and reconstructed matrices. In order to reconstruct the data based on the Given's parameters, we apply the procedure described in \Cref{prop:givens}. The average reconstruction error over 10 repeated experiments is $1.3^{-14}$. Therefore, we conclude that our theoretical attacks are indeed possible in realistic implementation.

\subsection{Linear regression}
\label{sec:experiments:linreg}
We implement the QR decomposition scheme and a prototype for linear regression in python to show that they provide accurate results in practice. In this experimental study we use three example data sets from \texttt{sklearn} and Kaggle: the Pima Indians diabetes \cite{diabetes}, WHO life expectancy \cite{who} and fish market \cite{fishmarket} data sets. We split the data sets horizontally in \num{5} chunks. We compute the baseline reduced QR decomposition using \texttt{scipy.linalg.qr}. As an error measure, we use the Frobenius norm between the centralized and federated $\Q$ and $\R$ matrices ($||\Q_c-\Q_f||_F$,$||\R_c-\R_f||_F$). For the linear regression, we use the \texttt{lm} function in \texttt{R} as a reference, as it uses QR decomposition as its standard solver. As additional error measures, we compute the sum of the absolute differences between the coefficients ($\sum_{s \in [S]}\x_c-\x_f$), $r^2$-values ($r^2_c-r^2_f$), and $p$-values ($\sum_{s \in [S]} (p_c-p_f)$). The results of these experiments are summarized in \Cref{tab:experiments}. The matrices, coefficients and $r^2$ values are identical, and there only minor variations in the $p-value$. 
	
	\begin{table}[htb]
	\centering
		\caption{Results of the experiments}
		\label{tab:experiments}
		\begin{tabular}{llll}
			\toprule
			Dataset & Diabetes& WHO & Fish market\\
			\midrule
			$||\Q_c-\Q_f||_F$ 		& $3.0e^{-14}$	& 	$1.1e^{-14}$	&$3.8e^{-13}$\\
			$||\R_c-\R_f||_F$ 		&	$1.7e^{-14}$&	$7.5e^{-7}$	&$9.3e^{-12}$\\
			$\x_c-\x_f$				& $2.23e^{-11}$	& $4.5e^{-12}$	&	$3.04e^{-11}$	\\
			$\sum_{d} (p_c-p_f)$ 	&	$0.003$		& $0.019$		& 	$0.029$	\\
			$r^2_c-r^2_f$			& $1.4e^{-17}$	& $0 $			&	$2.5e^{-15}$\\
			\bottomrule		
		\end{tabular}
	\end{table}
	
\subsection{Implementation \& Hardware}
The experiments were run on a standard laptop with $8$ CPUs and $16$ GB RAM. The algorithms were implemented in Python using \texttt{numpy} and \texttt{scipy}. The code can be found in the corresponding Gitlab repository at \url{https://github.com/AnneHartebrodt/federated-qr}. 

		\section{Discussion and future directions}
	\label{sec:qr:discussion}
	In this manuscript, we evaluate the most popular algorithms for QR decomposition with respect to their confidentiality in a federated context. As explained in \Cref{sec:qr:fed-householder} and \Cref{sec:qr:fed-givens}, Householder reflection and Givens rotation have immediate drawbacks that make them unsuitable to hybrid federated learning where the parameters are securely aggregated, because it is possible to extract the original data from the parameters. This makes the presented federated Gram-Schmidt QR algorithm the only algorithm which does not trivially expose the original data under the assumed federated setting. We argued that the parameters revealed during the federated Gram-Schmidt orthonormalization procedure contain no more information than the upper triangular matrix $\R$, and therefore fulfill our privacy specification of federated QR decomposition.
	
	In this article, we assume a hybrid federated learning setup, where the global parameters become known in clear text. This means, the results may only partially translate to systems which rely on encrypting the entire learning process under homomorphic encryption or computing the whole algorithm using more advanced secure multiparty computation schemes. These techniques are still expensive in practice \cite{Cho2018, Al-Rubaie2017}, but might be required to provide secure algorithms for Householder factorization and Givens rotation. If privacy is not a concern, detailed investigations of potential gains in transmission rounds would be required to find the most efficient QR scheme, most likely Givens algorithm according to our preliminary analysis.
	
	The investigation of information leakage associated with the parameters exchanged during the federated QR orthonormalization spins a cautionary tale. We showed that it is possible to reconstruct the input matrices, if they are upper triangular, solely from the exchanged parameters, because the first aggregate is technically not an aggregate and triggers a revealing cascade. This means that with fewer than three parties, even the clients could reconstruct the other participants' matrices. We showed that for upper triangular matrices privacy breaches are possible. Therefore, further investigations on other special types of matrices will be required.

	\section{Conclusion}
	\label{sec:qr:conclusion}
	This work presented federated implementations of three popular QR algorithms and investigated them with respect to their privacy, if they are deployed in a hybrid federated system with secure parameter aggregation. Popular applications of QR factorization, such as linear regression, can therefore be translated to the federated domain. The privacy must be considered carefully: The use of the outer vector product in Householder factorization, introduces a trivial confidentiality breach. Likewise, a trivial privacy leak in Givens rotation makes this algorithm unsuitable for the chosen federated learning paradigm. We come to the conclusion, that only Gram-Schmidt QR decomposition is suitable, due to it's reliance on inner vector products. Special matrices, such as upper triangular matrices, may still be more vulnerable to confidentiality breaches.

	\section*{Acknowledgements}
		The FeatureCloud project has received funding from the European Union’s Horizon 2020 research and innovation programme under grant agreement No 826078. This publication reflects only the authors’ view and the European Commission is not responsible for any use that may be made of the information it contains.
	\bibliographystyle{IEEEtran}
	\bibliography{references}
	

	\balance

\end{document}